\documentclass{article}

\usepackage{amsmath,amssymb}
\usepackage{amsthm}

\usepackage[usenames]{color}
\usepackage{rotating}
\usepackage{setspace}
 \newtheorem{definition}{Definition}%[section*]
 \newtheorem{proposition}[definition]{Proposition}
 \newtheorem{example}[definition]{Example}
 
 \newtheorem{corollary}[definition]{Corollary}
\newtheorem{lemma}[definition]{Lemma}

\begin{document}
%\begin{frontmatter}

%% \title{k-anonymity in graphs: compromizing efficiency and risk}
\title{A formalization of re-identification in terms of compatible probabilities}
\author{
%%\begin{flushleft}
  Vicen\c c Torra$^{1}$, Klara Stokes$^{2}$\\
\small $^1$ IIIA, Institut d'Investigaci\'o en Intel$\cdot$lig\`encia Artificial\\ 
\small  CSIC, Consejo Superior de Investigaciones Cient\'{\i}ficas\\ 
\small  Campus UAB, 08193 Bellaterra, Catalonia, Spain\\
\small  Email: vtorra@iiia.csic.es\\ 
\small $^2$ Universitat Oberta de Catalunya\\
\small  Email: kstokes@uoc.edu\\ 
%%\small $^2$ Universitat Rovira i Virgili\\
%%\small  Dept. of Computer Engineering and Maths, \\
%%\small  UNESCO Chair in Data Privacy\\
%%\small  Av. Pa\"{\i}sos Catalans 26, 43007 Tarragona, Catalonia, Spain\\
%%\small  Email: klara.stokes@urv.cat\\ 
%%}
%%\end{flushleft}
}

\maketitle

\begin{abstract}
Re-identification algorithms are used in data privacy to measure disclosure risk. They model the situation in which an adversary attacks a published database by means of linking the information of this adversary with the database. 

In this paper we formalize this type of algorithm in terms of true probabilities and compatible belief functions. The purpose of this work is to leave aside as re-identification algorithms those algorithms that do not satisfy a minimum requirement. 
\end{abstract}

\section{Introduction}
Privacy preserving data mining (PPDM) and statistical disclosure control (SDC) are two active areas of research that study how to avoid disclosure of sensitive information when data is released to third parties for their analysis. 

One of the existing approaches for ensuring privacy consists of manipulating the datafile adding some noise or reducing the quality of the information. Several data protection methods have been developed in this direction. Noise addition, microaggregation, rank swapping and PRAM are some of the existing methods. In general, these methods consists of transforming a data file $X$ by means of a masking method $\rho$ into a new data file $Y$. That is, the masking method returns $Y:=\rho(X)$. 

Methods reduce the disclosure risk at the expenses of some information loss. In other words, the results from an analysis of $X$ will in general give different results than an analysis of $Y$. With the aim of quantifying this loss, several information loss measures have been defined in the literature. 

Nevertheless, although the function $\rho$ inflicts a perturbation in the data that causes some information to be lost, the modification of the data may be insufficient for ensuring privacy. Due to this, disclosure risk measures have been defined and studied in the literature. 

Disclosure risk measures can be defined in terms of re-identification. This corresponds to identity disclosure. Re-identification algorithms permit us to model the situation in which an adversary wants to attack a published data set using some information that he has available. The adversary tries to link his information expressed as records in a datafile with the records in the published data set. The more records he reidentifies, the larger the risk. Therefore, given a particular file, the proportion of reidentified records is a measure of the risk. 

The concept of re-identification is also the cornerstone of the theory of $k$-anonymity. A dataset is $k$-anonymous if for each record in the dataset, there are other $k-1$ records that are equal to it. Nevertheless, as pointed out in~\cite{{ref:Stokes.Torra.2012:PAIS}}, the important question here is not whether the records have the same or different values, but that the records are indistinguishable in the re-identification process (when the adversary attacks the dataset). This idea permitted us in~\cite{ref:Stokes.Torra.2012:SOCO} to define $n$-confusion as an alternative to $k$-anonymity which provides the same level of anonymity without requiring records to have the same values. 

Because of that, re-identification algorithms are fundamental in data privacy and the literature presents several algorithms for re-identification~\cite{{ref:Torra.Abowd.Domingo.2006:PSD},{ref:Winkler.2004:PSD},{ref:Yancey.Winkler.Creecy.2002}}. The literature also discusses some models~\cite{{ref:Fellegi.Sunter.1969},{ref:Abril.Navarro-Arribas.Torra.2012:IF}} for re-identification that are used to determine the parameters of the algorithms. Nevertheless, up to our knowledge, there is no approach for how to formalize and determine correctness of re-identification algorithms. That is, there is no discussion on what a proper and correct re-identification algorithm is, and what kind of result a correct re-identification algorithm should give. 

%% In the standard scenario, we consider that the information the adversary has can be represented in terms of a database. Then, the adversary wants to link the data (e.g. records) in his data base with the data in the original population. The extent in which he is able to do so gives a measure of the risk. For example, in the literature we find as a measure of risk the proportion of records that can be linked to a protected file. 

In this paper we present a formalization of re-identification in terms of belief functions and true probabilities. 

The basic idea is that a good re-identification algorithm, given some information, a probability distribution over a population. If we assume that this re-identification algorithm behaves correctly, then it cannot return any probability distribution but must return a distribution that is compatible with the true one. In addition, we would expect that the more information we have available, the more the probability of the algorithm should resemble the true one. 

In this paper we model this situation in terms of belief functions~\cite{{ref:Shafer.1976}} and the transferable belief model~\cite{ref:Smets.Kennes.1994}. Departing from a true probability, we define two types of re-identification algorithms. First, we define a re-identification algorithm as one that returns a belief function that is compatible~\cite{ref:Chateauneuf.1994} with the true probability, and later as one that returns the pignistic transformation of a belief function that is compatible with the true probability. 

The structure of the paper is as follows. In Section~\ref{sec:2} we review some concepts that are needed later in this work. In particular, we discuss belief functions and re-identification algorithms. In Section~\ref{sec:4}, we introduce our model and discuss some relevant results about it. The paper finishes with some conclusions. 

\section{Preliminaries}
\label{sec:2}

This section is divided into three parts. In the first part we review some concepts related to belief functions, a model for approximate reasoning. We will use belief functions to construct a model for record linkage. In the second part we review k-anonymity, one of the approaches in data masking. 

\subsection{Belief functions}
\label{seccio.belief.functions}
Belief functions can be used to represent uncertainty with respect to probability distributions. We will not go into the details of their justification. The description in this section focuses on the concepts we need in the rest of the paper. For details and additional discussion see e.g.~\cite{{ref:Chateauneuf.1994},{ref:Smets.Kennes.1994},{ref:Walley.1991}}.

\begin{definition} 
\label{def:chMeas:Bel}
A set function $Bel: 2^{X} \rightarrow [0,1]$ is a {\em belief function} if and only if it satisfies 
\begin{enumerate}
\item[(i)] $\mu(\emptyset) = 0$, $\mu(X) = 1$ (boundary conditions)
\item[(ii)] $A \subseteq B$ implies $\mu(A) \leq \mu(B)$ (monotonicity)
\item[(iii)] For all $A_1, \dots, A_n \subseteq X$, 
\begin{eqnarray}
\label{eq:chMeas:Bel:ineq}
Bel(A_1 \cup ... \cup A_n) & \geq & \sum_j Bel(A_j) - \sum_{j < k} Bel(A_j \cap A_k) + ... + \nonumber \\
                           &      & (-1)^{n+1} Bel(A_1 \cap ... \cap A_n).
\end{eqnarray}
\end{enumerate}
\end{definition}

Belief functions are closely related to basic probability assignments. There is a basic probability assignment for each belief function, and a belief function for each basic probability assignment. 

\begin{definition}
A function $m:2^{X} \rightarrow [0,1]$ is a {\em basic probability assignment\index{basic probability assignment}} if and only if 
\begin{description}
\item[(i)] $m(\emptyset) = 0$
\item[(ii)] $\sum_{A \subseteq X} m(A) = 1$
\end{description}
\end{definition}

There exist two names for this function in the literature: basic probability assignment (e.g. in~\cite{ref:Shafer.1976}) and basic belief assignment (e.g. in~\cite{ref:Smets.Kennes.1994}). In the rest of the paper we will say just assignment. 

The following proposition establishes the relationship between assignments and belief functions. 

\begin{proposition}
Let $Bel$ be a belief function defined on the reference set $X$, then the function $m_{\mu}$ defined below is a basic probability assignment $m$: 
\begin{equation}
\label{eq:chMeas:mFromBel}
m_{\mu}(A) = \sum_{B \subseteq A} (-1)^{|A|-|B|} Bel(B)  \qquad for~all~ A \subseteq X.
\end{equation}
Let $m$ be a basic probability assignment, then, the function $Bel_{m}$ defined below is a belief function
\begin{equation}
\label{eq:chMeas:belFromM}
Bel_{m}(A) := \sum_{B \subseteq A} m(B) \qquad for~all~ A \subseteq X.
\end{equation}
\end{proposition}

%% Then, $m(A)$, the mass given to the set $A$, represents our support to $A$ without supporting a more specific subset of $A$. 
%% Given $m: 2^{X} \rightarrow [0,1]$, the degree of belief $Bel$ of a solution in $A$ is defined by: 
%% $$Bel(A)=\sum_{B\subset A} m(B).$$

Belief functions generalize probabilities. In particular, when $m(A)=0$ for all $A$ such that $|A|>1$, then $Bel$ is a probability. In this case, the assignment $m$ to the singletons is the probability distribution. That is, $P(\{x\})=m(\{x\})$ for all $x\in X$, and then $P(A)=Bel(A)$ for all $A \subseteq X$. 

As stated before, belief functions can represent uncertainty in probability distributions. They permit to differentiate situations which standard probabilities cannot. For example, total ignorance in a set $X$ is modeled defining $m(X)=1$ and $m(A)=0$ for all $A \neq X$. In contrast, when we know that the elements in $X$ all have the same support we assign $m(x)=1/|X|$ for all $x \in X$. Note that this is different from the case of standard probabilities where both situations are represented by $P(x)=1/|X|$ for $x \in X$. 

Given a belief function $Bel$ defined from $m$, Dempster defined the pignistic transformation as a function that finds a probability distribution from $Bel$. This pignistic transformation is based on the transferable belief model by Smets~\cite{ref:Smets.Kennes.1994} that distinguishes between the credal and the pignistic level. The credal level is where beliefs are taken into consideration and operated on, and the pignistic level is where beliefs are used. Although we do not understand probabilities and beliefs as subjective, as Smets does, both levels are appropriate for modeling re-identification. An ideal re-identification algorithm will compute belief functions in the credal level, with the minimal possible commitments in case of uncertainty. Then, when decisions are to be made, we move to the pignistic level and probabilities are made concrete. 

\begin{definition}
\label{def:pignisticTransformation}
Let $Bel$ be a belief function, then we define the pignistic probability distribution derived from $Bel$, $P_{Bel}$, as: 
$$P_{Bel}(\{x\}) = \sum_{\{B \in 2^X: x \in B\}} \frac{m(B)}{|B|}$$
\end{definition}

\subsection{Data protection methods}
Formally, given a data set $X$, a masking method $\rho$ constructs a data set $Y:=\rho(X)$. Data privacy studies masking methods that return datasets which can be released to third parties in a way that avoids disclosure of sensitive information, but preserves the value of the data as material for analysis. 

One of the existing concepts for data privacy is $k$-anonymity. A dataset satisfies $k$-anonymity when for each record there are $k-1$ other records that are indistinguishable in the dataset. 

Several algorithms have been proposed in the literature to build a dataset compliant with $k$-anonymity through generalization, suppression and clustering. For example, if we have 6 records with values 18,16,19,22,24,24 for attribute $V_1$, we can consider the intervals [15,19],[20,25] and then recode the original values according to these intervals. Doing so, we will have three in the interval [15,19] and three other in the interval [20,25]. This ensures $k$-anonymity for $k=3$ if only this sole attribute is discussed. 

We will use $gen_{V_i}$ to denote a method that ensures $k$-anonymity for a single attribute $V_i$ using generalization for some appropriate value $k$. 

\section{Re-identification algorithms}
%%\label{sec:3}

Given a dataset $X$, a protection method $\rho$, and the protected dataset $Y:=\rho(X)$, disclosure risk can be measured in terms of the number of records in $Y$ that can be correctly reidentified. Indeed, a common approach when constructing re-identification algorithms is to optimize with respect to this criterion~\cite{ref:Abril.Navarro-Arribas.Torra.2012:IF}. Nevertheless, although formalizations of the expected outcome of these algorithms exist (i.e., we expect a method to maximize the number of correct links), no formalization exists on what we mean when we say that a re-identification method is correct. We would like a formalization that excludes re-identification methods which perform incorrect re-identifications. In this paper, we discuss a formalization based on belief functions and a true probability. We will base our discussion on a previous definition of re-identification algorithms used to define $n$-confusion in ~\cite{ref:Stokes.Torra.2012:PAIS}. The definition is as follows. 

\begin{definition}
\label{def:reidentification}
\cite{ref:Stokes.Torra.2012:PAIS}
Let $\rho$ be a method for anonymization of databases, $X$ a table with $n$ records indexed by $I$ in the space of tables $D$ and $Y:=\rho(X)$ the anonymization of $X$ using $\rho$. 
Then a re-identification method is a function that, 
given a collection of entries $y$ in $\mathcal{P}(Y)$ and some additional information from a space of auxiliary informations $A$, 
returns the probability that $y$ are entries from the record with index $i\in I$,
$$\begin{array}{rccc}r:&\mathcal{P}(Y) \times A &\rightarrow &[0,1]^{n}\\\\
&(y,a) &\mapsto &\left(P(y \textrm{ corresponds to record } X[i]): i\in I\right).
\end{array}$$
%cut here
Consider the objective probability distribution corresponding to the re-identification problem. Then, we require from a re-identification method that it returns a probability distribution that is compatible with this probability, also when missing some relevant information. Compatibility can be modeled in terms of compatibility of belief functions (see \cite{{ref:Chateauneuf.1994},{ref:Smets.Kennes.1994}}). 
\end{definition}

Section~\ref{sec:4} discusses re-identification algorithms and the compatibility issue mentioned above. Before, we review some of the approaches that can be found in the literature on record linkage. Recall that we have defined record linkage in terms of the probability that $y$, the protected record, are entries from the record with index $i \in I$, and we denote this by $r(y,a)[i]$. 

As we will see, for some methods we can understand probabilities as following a Bayesian objective approach, and for other methods as subjective probabilities. In the latter case, we can understand the probabilities as votes. 

Some of the methods described below are not completely formalized in the literature, so the formulation is ours.

\begin{itemize}
\item {\bf $k$-anonymity and re-identification.} Re-identification methods applied to $k$-anonymous databases return for each $y$ a list of possible records ${\cal J} \subset I$ that are possible matches of the given record. 

On the basis of the principle of insufficient reason (or principle of indifference) this can be modeled by means of a uniform distribution over the records in ${\cal J}$. That is, $r(y,a)[i] = 1/|{\cal J}|$ for all $i \in {\cal J}$. 

An alternative model, not previously considered in the literature up to our knowledge, is to consider belief functions. This will be the subject of this article, and Example~\ref{ex:kAnon.probs.belief} focuses on the use of belief functions for their use in $k$-anonymity. 

\item {\bf Probabilistic record linkage.} The mathematical model formalized by Fellegi and Sunter in 1969~\cite{ref:Fellegi.Sunter.1969} is based on a probabilistic model that computes the probability of a particular coincidence pattern $\gamma$ conditioned by the existence of a match: $P(\gamma|Match)$. Probabilistic record linkage returns the probability of a correct match given a particular coincidence pattern (i.e, $P(Match|\gamma)$). The Bayes' rule is used in this process. This situation can be modeled by 
$$r(y,a)[i] = P(Match|\gamma(y,x_i)).$$ 

\item {\bf Specific attacks to data protection methods.} The approach to attack rank swapping in~\cite{ref:Nin.Herranz.Torra.2008:DKE.rankSwapping} can be represented in terms of a list of candidates, in the line of the re-identification methods for $k$-anonymity as described above. Re-identification attacks for rank swapping $p$-buckets can be modeled by means of probabilities. 

\item {\bf Distance-based record linkage.} Some literature exists where re-identification methods assign to each record in one file the most similar record (at a minimum distance) in the other file. In this case, probabilities can be defined from the distances, but such assignments should typically be only interpreted as voting or indications for subjective probabilities. \cite{ref:Copas.Hilton.1990} is an exception to this, where a real probability is estimated taking into account the similarity between records. 
\end{itemize}

\section{A formalization of re-identification}
\label{sec:4}

In this section we analyse the concept of re-identification further. In Definition~\ref{def:reidentification}, re-identification is a function that, given some partial information on $Y$ and some additional information, returns a probability distribution on the set of records. 

We claim that this probability distribution should be compatible with the true probability. Our motivation is to create a theoretical foundation for disclosure risk evaluation. The formalization will leave aside those re-identification algorithms that do not satisfy some minimum requirements. In particular, we will not approve algorithms that deliver incorrect results on purpose, and we will force algorithms to perform as well as possible, according to the evidence found in the data and any a priori knowledge. For the purpose of risk evaluation, using the worst case scenario, this implies no loss of generality. 

In this section, we introduce a formalization that relies on a true probability of re-identification. This true probability corresponds to the case in which we know {\em everything} about the whole anonymization process, and assuming this information is used in the re-identification process. That is, the true probability only includes the uncertainty that cannot be removed because e.g. randomness. 

We would expect that the re-identification process leads to a probability that is less informative than the true probability in case of uncertainty, e.g. on the masking process or on the data available for re-identification. Examples of such uncertainty could be that some variables are not included in the risk analysis, or that part of the masking process is not disclosed and cannot be taken into account in the risk analysis. 

%% Nevertheless, such uncertainty cannot lead to arbitrary probability distributions. 
Nevertheless, uncertainty does not justify all probability distributions. 
Only some of them are valid. As an extreme example, we cannot accept as a re-identification method one that assigns $r(y,A)[i]=1$ if and only if $i = i_0$ for any $y \in Y$. In order to represent less informative probabilities we use imprecise probabilities and, more specifically, belief functions. As stated in Section~\ref{seccio.belief.functions}, belief functions can be used when there is uncertainty in the values of a probability distribution. When no additional uncertainty is present in the re-identification process, the corresponding belief function is equivalent to a probability distribution. 

We will pressume that an ideal re-identification method is the one that expresses uncertainty by means of a belief function. The belief function computed by this re-identification method should be {\em compatible} with the true probability. 

Here we use the term {\em compatible} according to Chateauneuf~\cite{ref:Chateauneuf.1994}, who defined it for belief functions. Definition 1 in~\cite{ref:Chateauneuf.1994} defines two belief functions as compatible when the joint information is non-empty. The definition which we will use here is the same as Chateauneuf's definition except for the fact that we will compare a probabilty (the true one) and a belief function. 

We will use $P$ to denote the true probability of re-identification. We give its formal definition below. 

\begin{definition}
Let $X$ be a dataset, $\rho$ a data masking method, and $Y:=\rho(X)$. Then, we define the true probability $P_{\rho,X,Y}(x_i | y_i)$ as the probability that the protected record $y_i$ proceeds from the record $x_i$ given $\rho$, $X$, and $Y$. 
\end{definition}

Given a true probability and a belief function, we define their compatibility as follows. 

\begin{definition}
Given a probability $P$, we say that a belief function $Bel$ is compatible with $P$ if $P \geq Bel$. 
\end{definition}

For the sake of illustration, let us consider the following example with partial information in the re-identification process. This will be a running example of this paper. 

\begin{example}
\label{ex:kAnon.probs}
Let $X$ be a dataset with different attributes $V_1, \dots, V_m$. Consider the masking method $\rho$ where each attribute $V_i$ is protected by means of a generalization method $gen_{V_i}$ which ensures $k$-anonymity for $k=k_i$. That is, given that $X[V_i]$ represents the column of $X$ with attribute $V_i$, we have that $gen_{V_i}$ is applied to $X[V_i]$ for all $i \in \{1, \dots, m\}$, and $Y$ is defined in terms of the results of $gen_{V_i}$ putting their results side by side as follows 
$$Y:=\rho(X)=\left[gen_{V_1}(X[V_1])||\dots||gen_{V_m}(X[V_m])\right].$$

The true probability $P_{\rho,X,Y}$ for the re-identification of $X$ and $Y:=\rho(X)$ for a given record $y \in Y$ assigns the same non-zero probability to all records $x$ in $X$ such that $y$ can proceed from $x$ (taking into account the generalization processes $gen_{V_i}$), and assigns 0 to all other records. Formally, let the record $y$ be $y=(y_1, \dots, y_m)$ and let us define the candidate set of $y$ as the records $x \in X$ such that $y$ can proceed from $x$ (i.e., $CandidateSet(y)=\{x| y = (gen_{V_1}(x_1), \dots, gen_{V_m}(x_m))\}$). Then, the true probability of $x$ given $y$ is defined by: 

\begin{displaymath}
  P_{\rho,X,Y}(x | y) = \left\{
  \begin{array}{ll}
    \frac{1}{|CandidateSet(y)|} & \textrm{if~} x \in    CandidateSet(y) \\
    0                           & \textrm{if~} x \notin CandidateSet(y) \\
\end{array} \right.
\end{displaymath}

\end{example}
Re-identification methods that are applied to subsets of $Y$ consisting of only some attributes will lead to probability distributions that may be different from the true probability distribution. If this is the case, then they will be less informative. Note that, when only a subset of attributes $V' \subseteq \{V_1, \dots, V_m\}$ are considered, then the re-identification algorithm may select more candidates than there are in the true candidate set. 

In the next example we consider the re-identification of the $i$th register of $Y$ taking into account only partial knowledge consisting of some of the attributes. 

\begin{example}
\label{ex:kAnon.probs.reidentif}
Let $X$, $\rho$, $Y:=\rho(X)$, $y \in Y$ and $gen_{V_i}$ be defined as in Example~\ref{ex:kAnon.probs}. Let $Attrs(j)$ for $j=1, \dots, 2^m$ represent all possible (non-empty) subsets of attributes of $V=\{V_1, \dots, V_m\}$ indexed by $j$. Let $y_{j} \subseteq \mathcal{P}(Y)$ for $j=1, \dots, 2^m$ represent a record of the database $Y$ restricted to $Attrs(j)$. An example of an indexation of attribute subset is $Attrs(1)=\{V_1\}$ and $Attrs(3)=\{V_1, V_2\}$. 

Then, we expect a re-identification method applied to $y_{j}$, and taking into account how $Y$ is generated from $X$ using $\rho$, to deliver the following probability distribution: 

\begin{displaymath}
  r(y_{j},a)[i] = \left\{
  \begin{array}{ll}
    \frac{1}{|CandidateSet_{Attrs(j)}(y)|} & \textrm{if~} x_i \in    CandidateSet_{Attrs(j)}(y) \\
    0                                      & \textrm{if~} x_i \notin CandidateSet_{Attrs(j)}(y) \\
\end{array} \right.
\end{displaymath}
where $CandidateSets_{A}(y_j)$ includes a record $x \in X$ if $y$ can proceed from $x$ when only the attributes in $A$ are considered. 
\end{example}

It is easy to see that for any record $y$ in $Y$, we have 
$$CandidateSet(y) \subseteq CandidateSet_{Attrs(j)}(y)$$ 
and therefore $P_{\rho,X,Y}(x_i|y) \geq r(y,a)[i]$ for all $x_i$ in $X$. 

In this example the re-identification method assigns to all records in the candidate set the same probability. This is the usual way to assign probabilities under the principle of indifference. Nevertheless, in reality we only know that the true match is one of the records in the candidate set and that we do not have any preference on them. If we instead allow the re-identification to return a belief function, then this situation can be properly expressed. We define now the concept of {\em re-identification method expressing uncertainty}, a re-identification method that is not required to assign probabilities to singletons.

%% We expect that the probability distributions $r$ constructed in Example~\ref{ex:kAnon.probs} are compatible with the true probability $P$ in the same example. We want to underline that the difference between the true probabilities and the ones inferred using the reidentification method are caused by the uncertainty of missing data. 
%%
%%We will now discuss the use of belief functions to model this uncertainty. Belief functions permit to model uncertainty on a probability distribution. We will use them to model the partial information that the reidentification methods have. To do so, we consider a more general setting in which reidentification methods lead to belief functions instead of probabilities. In this case, the belief functions are assumed to be compatible with the true probabilities. 
%%
%% TRETA DEFINICIO 8
%%
%%This definition is based on the definition of belief functions as lower bounds on some unknown true probability, following~\cite{ref:Chateauneuf.1994}. As belief functions are better tools to represent uncertainty than probability distributions, we define a reidentification algorithm expressing uncertainty as follows. 

\begin{definition}
\label{def:reidentification.amb.belief}
Let $\rho$ be a method for anonymization of databases, $X$ a table with $n$ records indexed by $I$ in the space of tables $D$ and $Y:=\rho(X)$ the anonymization of $X$ using $\rho$. Let $P_{\rho,X,Y}(x | y_i)$ be the true probability of $\rho$, $X$ and $Y$. 
Then a {\em re-identification method expressing uncertainty} is a function that, given a collection of entries $y$ in $\mathcal{P}(Y)$ and some additional information from a space of auxiliary informations $A$, returns the belief function compatible with the true probability $P_{\rho,X,Y}(x_i | y)$ that $y$ are entries from the record with index $i \in I$
$$\begin{array}{rccc}r^*:&\mathcal{P}(Y) \times A &\rightarrow &[0,1]^{2^{|X|}}\\\\
&(y,a) &\mapsto &\left(m(y\textrm{ proceeds from a record in }B): B \subseteq X\right).
\end{array}$$
%% For each $y \in Y$, the belief function $Bel(y,a)$ induced from $m$ has to be compatible with the true probability $P_{\rho,X,Y}(x_i | y)$. 
\end{definition}

%% Consider the objective probability distribution corresponding to the re-identification problem. Then, we require from a re-identification method that it returns a probability distribution that is compatible with this probability, also after missing some relevant information. Compatibility can be modeled in terms of compatibility of belief functions (see \cite{{ref:Chateauneuf.1994},{ref:Smets.Kenn

As for any belief function, in this definition we expect 
\begin{description}
\item [(1)] $m(X)=1$ and $m(A)=0$ for all $A \neq X$ when there is no evidence on which are the original records corresponding to the protected record $y$, and 
\item [(2)] an increment of the belief function for $B \subseteq X$ when the evidence increases for records in $B$. 
\end{description}
In addition, the same belief functions will apply to different protected records whenever these have the same values. Formally, we have that $Bel(y,a)=Bel(y',a)$ if $y=y'$ holds. 

The following example illustrates the use of a re-identification method expressing uncertainty. 

\begin{example}
\label{ex:kAnon.probs.belief}
Let $X$, $\rho$, $Y:=\rho(X)$, $y \in Y$, $y_j$, $gen_{V_i}$, $CandidateSet$ and $CandidateSet_{A}$ be defined as in Examples~\ref{ex:kAnon.probs} and~\ref{ex:kAnon.probs.reidentif}. Then, the belief function $r^*(y_j,a)$ that better represents the uncertainty is defined by the following assignment: 

\begin{displaymath}
  m(A) = \left\{
  \begin{array}{ll}
    1                                      & \textrm{if~} A = CandidateSet_{Attrs(j)}(y) \\
    0                                      & \textrm{otherwise} \\
\end{array} \right.
\end{displaymath}

Therefore, for all $B \subseteq X$, 
$$r^*(y_{j},a)(B) = \sum_{A \subseteq B} m(A)$$

It is easy to see that $r^*(y_j,a)(B)=1$ if and only if $CandidateSet_{Attrs(j)}(y) \subseteq B$. 
\end{example}

For the belief functions in this example we can prove the following.

\begin{proposition}
The belief functions $r^*(y_j,a)$ defined in Example~\ref{ex:kAnon.probs.belief} are compatible with the true probability in Example~\ref{ex:kAnon.probs}. 
\end{proposition}

\begin{proof}
For simplicity, let us use the notation $Bel(B)=r^*(y_j,a)(B)$. We need to prove that $P(C) \geq Bel(C)$ for all $C \subseteq X$. Since $Bel(B) \in \{0,1\}$, we only need to check two cases. 
\begin{itemize}
\item When $Bel(C)=0$, it is clear that $P(C) \geq Bel(C)$ for all $C$. 
\item When $Bel(C)=1$, then $C \supseteq CandidateSet_{Attrs(j)}(y) \supseteq CandidateSet(y)$. Therefore, $P(C) \geq P(CandidateSet_{Attrs(j)}(y)) \geq P(CandidateSet(y)) = 1$. So, $P(C)=Bel(C)$. 
\end{itemize}
In the case of $Bel(C)=1$ we use the condition discussed above that 
$$CandidateSet(y) \subseteq CandidateSet_{Attrs(j)}(y).$$

This proves the proposition. 
\end{proof}

In contrast, if the re-identification method assigns $m(B)=1$ to a set $B$ that misses one record $x_i$ of the candidate set of $y$, then the inferred belief function is not compatible with the true probability. This is formalized in the next lemma. 

\begin{lemma}
Let $x_0$ be a record of the candidate set, let $B$ be an arbitrary subset of $X$, and let 
$$C_0=(B \cup CandidateSet(y)) \setminus \{x_0\}.$$ 
Let a re-identification method assign $m(C_0)=1$ and $m(C) = 0$ for all $C \subseteq X$ such that $C \neq C_0$. The belief function induced from $m$ is not compatible with the true probability. 
\end{lemma}

\begin{proof}
It is easy to see that the belief function satisfies $Bel(C_0)=1$. Nevertheless, since $C_0$ does not include $x_0$, we have that the true probabiliy for $C_0$ is 
$$P(C_0) = 1 - \frac{1}{|CandidateSet(y)|}.$$
Therefore, as $P(C_0) < Bel(C_0)$, the belief function is not compatible with the probability. 
\end{proof}

It is important to note that this result removes from the set of valid re-identification methods expressing uncertainty those that miss the correct records from the candidate set. 

\subsection{Pignistic transformation and re-identification methods}
In the previous section we defined a general re-identification method that returns a belief function, and this belief function is required to be compatible with the true probability. Nevertheless, as re-identification methods in real applications return probabilities, we reconsider our definition of re-identification algorithms so that they also return probability distributions. Nevertheless, these probability distributions are required to proceed from the belief function. 

In particular, the probability is constructed from the belief function following the principle of insufficient reason (or principle of indifference). That is, the assignment $m$ to a set is distributed to the singletons of this set according to a uniform distribution. We say that a probability constructed in this way is compatible with the original distribution. 

This construction precisely corresponds to the {\em pignistic transformation} and follows the transferable belief model by Smets. Details of a characterization of the transformation is given in~\cite{ref:Smets.Kennes.1994}. This pignistic transformation was defined in Definition~\ref{def:pignisticTransformation}. 

\begin{definition}
Given two probabilities $P$ and $P'$, we say that $P'$ is compatible with $P$ if there exists a belief function $Bel$ compatible with $P$ such that $P'$ is the pignistic probability distribution derived from $Bel$ (i.e., $P'=P_{Bel}$). 
\end{definition}

We now present the pignistic probabilities for the running example. 

\begin{example}
The pignistic probabilities for the belief functions of Example~\ref{ex:kAnon.probs.belief} for $y_j$ are as follows: 

\begin{displaymath}
  P(y_j,a)[i] = \left\{
  \begin{array}{ll}
    \frac{1}{|CandidateSet_{Attrs(j)}(y)|}                                     & \textrm{if~} x_i \in CandidateSet_{Attrs(j)}(y) \\
    0                                                                          & \textrm{otherwise.} \\
\end{array} \right.
\end{displaymath}
\end{example}

The probabilities defined in this example satisfy the following inequalities: 

\begin{itemize}
\item If $x_i \in CandidateSet(y)$, then $x_i \in CandidateSet_{Attrs(j)}(y)$ for all $j$. Then, we have 
$$P(y_j,a)[i]=\frac{1}{|CandidateSet_{Attrs(j)}(y)|} \leq P_{\rho,X,Y}(x_i|y) = \frac{1}{|CandidateSet(y)|}$$
\item If $x_i \notin CandidateSet(y)$ and $x_i \in CandidateSet_{Attrs(j)}(y)$, then we have 
$$P(y_j,a)[i]=\frac{1}{|CandidateSet_{Attrs(j)}(y)|} \geq P_{\rho,X,Y}(x_i|y) = 0$$
\item If $x_i \notin CandidateSet(y)$ and $x_i \notin CandidateSet_{Attrs(j)}(y)$, then we have
$$P(y_j,a)[i] = P_{\rho,X,Y}(x_i|y) = 0$$
\end{itemize}

Therefore, in general, for a given $x_i$ the pignistic probability can be either larger or smaller than the true probability. However, it cannot assign a zero probability to an $x_i$ in the candidate set. 

In fact, the next proposition proves that the support of the probability of the re-identification method should contain the support of the true probability. 

\begin{proposition}
\label{prop.16}
Let $P$ be a true probability, let $Bel$ be a belief function compatible with $P$, let $P'$ be the pignistic probability derived from $Bel$. Let $B_{P}$ be the support of $P$ (i.e., $B_P=\{x | P(x)\neq 0\}$), and let $B_{P'}$ be the support of $P'$. Then, $B_{P} \subseteq B_{P'}$. 
\end{proposition}
\begin{proof}
Let $x_0$ be an arbitrary element of the support of $P$, so $P(x_0)\neq 0$. Let $\alpha=1-P(x_0)$. Then, taking into account that $Bel$ is compatible with $P$, we have
$$1> \alpha = 1-P(x_0)=P(X \setminus \{x_0\}) \geq Bel(X \setminus \{x_0\}) = \sum_{C \subseteq X \setminus \{x_0\}} m(C). $$
So, 
$$1 - \sum_{C \subseteq X \setminus \{x_0\}} m(C) > 0.$$
As $\sum_{C \subseteq X} m(C) = 1$, we have
$$\sum_{C \subseteq X: x_0 \in C} m(C) = 1 - \sum_{C \subseteq X \setminus \{x_0\}} m(C) > 0.$$
So, there exists at least one $C$ such that $x_0 \in C$ and $m(C)\neq 0$. 

Then, by definition of the pignistic transformation, for all $x \in B$, $P'(x)\neq 0$. Therefore, $x_0$ is an element of the support of $P'$. As $x_0$ is an arbitrary element of the support of $P$, the statement is proven. 
\end{proof}

Now we will discuss two properties of the probability of re-identification that concern the case in which the probability is one for a single record. That is, we have that the true probability is a Dirac delta distribution at a single record $x_0$. This distribution is denoted by $\delta(x_0)$ and its value is 1 if and only if $x=x_0$. Note that this case is possible in Example~\ref{ex:kAnon.probs} when the intersection of the candidate sets of two (or more variables) is a singleton. Formally, $|CandidateSet_{V_i}(y) \cap CandidateSet_{V_j}(y)|=1$. A similar situation was exploited in~\cite{ref:Nin.Herranz.Torra.2008:DKE.rankSwapping} to attack rank swapping. 

\begin{lemma}
\label{lemma.bel.i.p.singleton}
Let $P$ be a Dirac delta distribution at $x_0$. Let $Bel$ be a belief function compatible with $P$. Then, $m(A)=0$ if any only if $A \cap \{x_0\}=\emptyset$. 
\end{lemma}

\begin{proof}
Suppose that $A \cap \{x_0\}=\emptyset$, then we have $P(A)=0$. 

As $P(A)\geq Bel(A)$, then we have $Bel(A)=0$. 

Therefore, as $Bel(A)=0=\sum_{B \subset A} m(B)$ and $m(B) \geq 0$ for all $B \subseteq A$, we have $m(A)=0$. 

The fact that $m(A)=0$ implies $A \cap \{x_0\}=0$ is a corollary of Proposition~\ref{prop.16}. 
\end{proof}

\begin{proposition}
\label{prop:P.singleton.llavors.P.peak}
Let $P(i)=r(y,a)$ be a Dirac delta distribution at $i_0$. Let $P'$ be a probability compatible with $P$. Then, if $P'$ has its maximum in $i'$ (i.e. $i' = \arg \max P'(i)$), then $i'=i_0$. 
\end{proposition}

%% that is, it has a peak in $i'$

\begin{proof}
Suppose that $P(i)=0$ for all $i \neq i_0$, and $P(i_0)=1$. Lemma~\ref{lemma.bel.i.p.singleton} implies that $m(A)=0$ for all $A \cap \{i_0\}=\emptyset$. Let ${\cal A}=\{A | m(A) \neq \emptyset\}$, then for any belief function compatible with $P$, 
$$P_{Bel}(\{i_0\})=\sum_{A: i_0 \in A} \frac{m(A)}{|A|} = \sum_{A \in {\cal A}} \frac{m(A)}{|A|} \geq \sum_{A \in {\cal A} : A\cap\{i_1\}=\{i_1\}} \frac{m(A)}{|A|} = P_{Bel}(\{i_1\}) $$
for all $i_1 \neq i_0$. 

If $P'$ is compatible with $P$, then there exists a belief function $Bel$ compatible with $P$ such that $P'=P_{Bel}$. 
\end{proof}

Note that this proposition is valid only when the true probability is a Dirac delta distribution. However this should not usually be the case if the data protection algorithm is effective. For example, in Example~\ref{ex:kAnon.probs} there may be $y$ with $|CandidateSet(y)|=1$, but for other $|CandidateSet(y)|=k>1$, so that the probability is $1/k < 1$. In general, we might even have that the record with a maximal probability is not one of the records in the candidate set. The next proposition establishes this fact. 

\begin{proposition}
\label{prop:maxim.p.no.solucio}
Let $X$ be a reference set with $|X|\geq 3$, let $A \subseteq X$ be a set of $k \geq 2$ records with a true probability for record $y$ equal to $1/k$. Then it is possible to have a probability compatible with the true probability such that the record with maximum probability is none of the ones in $A$. 
\end{proposition}

\begin{proof}
Let $P$ represent the true probability. Then $P(x)=1/k$ for all $x \in A$. 

Consider a record $x_0$ not in $A$. Therefore $P(x_0)=0$. Then, define a belief function $Bel$ in terms of $m$ as follows: 

\begin{displaymath}
  m(C) = \left\{
  \begin{array}{ll}
    \frac{1}{k}                 & \textrm{if~} C = \{x_0,x\} \textrm{ for any } x \in A \\
    0                           & \textrm{otherwise} \\
\end{array} \right.
\end{displaymath}

First we prove that this belief function is compatible with $P$. To do so, we need to prove that $P(B) \geq Bel(B)$ for all $B \subset X$. To do so, we consider two cases for the sets $B \subseteq X$ according to the membership of $x_0$ to $B$. Note that in both cases we have $P(B)=|B \cap A| \cdot (1/k)$. 

\begin{itemize}
\item Case $x_0 \in B$: As $Bel(B)=|B \cap A| \cdot (1/k)$, we have that $P(B) = Bel(B)$. 
\item Case $x_0 \notin B$: As $Bel(B)=0$, we have that $P(B) \geq Bel(B)=0$. 
\end{itemize}

Now we consider the pignistic probability from $Bel$. It is easy to prove that 

\begin{displaymath}
  P_{Bel}(x) = \left\{
    \begin{array}{ll}
      \frac{1/k}{2}                 & \textrm{if~} x \in A \\
      1/2                           & \textrm{if~} x = x_0 \\	 
      0                             & \textrm{otherwise} \\
\end{array} \right.		
\end{displaymath}

Therefore, we have that $x_0$ is the record with maximum pignistic probability when precisely $x_0$ is not in $A$. 
\end{proof}

%% This proposition establishes that, in general, we cannot ensure that the record with maximal probability according to a correct re-identification algorithm is the correct match. 

This proposition implies the following corollary. 

\begin{corollary}
Given a compatible probability $P'$ of a true probability $P$, the record with maximal value in $P'$ can be different from the record with maximal value in $P$. 
\end{corollary}

The next example illustrates a masking method that leads to a belief function and a probability distribution as the one in the above Proposition~\ref{prop:maxim.p.no.solucio}. 

\begin{example}
Let $\rho$ be a masking method defined as follows for records $x$ in ${\mathbb{N}}^3$. 
\begin{enumerate}
\item Let $\alpha$ be a random number in $\{0,1\}$ according to a uniform distribution. 
\item Let $\beta$ be a random number in $\{1,2,3\}$ according to a uniform distribution. 
\item Let $y:=x + \alpha e_{\beta}$ where $e_i$ is the unit vector in ${\mathbb{N}}^3$.
\end{enumerate}
Given $A=\{x_0=(0 0 0), x_1=(1 0 0), x_2=(0 1 0), x_3=(0 0 1)\}$ and $X \supseteq A$, we can model the re-identification of $y=(0 0 0)$ by means of the belief function in Proposition~\ref{prop:maxim.p.no.solucio}. Therefore, when we guess by selecting the most probable record using the pignistic transformation, we select $x_0$. However, if $\alpha$ is known to be zero, $x_0$ is impossible. 
\end{example}

The results given in this section describe the behaviour of our formalization for re-identification. At the same time, they give constraints on what we consider to be a proper re-identification algorithm, and, thus, they define the minimal requirements for these algorithms. 

\subsection{Evidence and uncertainty measures}
When new information is given to the re-identification algorithm, the belief function is updated according to this new evidence. The most particular case is when we consider that mass is transferred to a set $C_1$ from a larger set $C_2$. That is, we increment the mass of $C_1$ while reducing the one of $C_2$ and not modifying the rest of sets. 

The literature presents several definitions of uncertainty measures to evaluate either belief functions or probability distributions. Klir and Wierman~\cite{ref:Klir.Wierman.1999} give an account and a classification of some of these measures. 

In this section, we first show with examples that the entropy of the pignistic probability is not monotonic. We give these examples because entropy is often interpreted as a measure of information and, as such, one might think that in our case $C_1$ is more informative than $C_2$. As the examples show, in some cases the entropy is monotonic to this type of transformations, but in other cases it is not. 

Later, we prove that the measure of nonspecificity is monotonic with respect to the changes caused by transferring evidence from $C_2$ to $C_1$. This measure was defined by Dubois and Prade~\cite{ref:Dubois.Prade.1985:nonspecificity} as a generalization of the measure by Higashi and Klir~\cite{ref:Higashi.Klir.1983}. A characterization of this measure was given by Ramer in~\cite{ref:Ramer.1987} (see Klir and Wierman~\cite{ref:Klir.Wierman.1999} for details). We consider that this measure represents better the quantity of information in a belief function. 

The section finishes with a discussion on the uncertainty measures for re-identification. 

\begin{definition} 
The entropy of a belief function $Bel$ is defined as the entropy of the pignistic probability distribution derived from $Bel$. That is, 
\begin{equation}
Entropy(Bel)=\sum_{x \in X} P_{Bel}(x) \log P_{Bel}(x) 
\end{equation}
\end{definition}

Let us now consider two examples of the entropy of belief functions. 

\begin{example}
Let $X=\{x_1, \dots, x_8\}$ and let $Bel$ be the belief function defined by 
\begin{itemize}
\item $m(\{x_1, \dots, x_5\})=0.07692307 \cdot 5$
\item $m(\{x_1, \dots, x_8\})=0.07692307 \cdot 8$
\end{itemize}
The pignistic probability $P_{Bel}$ corresponds to: 
\begin{itemize}
\item $P_{Bel}(x_1)=P_{Bel}(x_2)=P_{Bel}(x_3)=P_{Bel}(x_4)=P_{Bel}(x_5)=0.15384614$
\item $P_{Bel}(x_6)=P_{Bel}(x_7)=P_{Bel}(x_8)=0.07692307$
\end{itemize}

Define $Bel'$ by {\em transferring} mass from $C_2=\{x_1, \dots, x_8\}$ to $C_1=\{x_1,x_2\}$. We have that, $C_1 \subseteq C_2$, and, therefore, $C_1$ is more specific than $C_2$. Let be the {\em transferred} mass be equal to $\Delta= 0.038461544 \cdot 8$. Therefore, we have that the new belief function is defined by: 
\begin{itemize}
\item $m'(\{x_1, x_2\})=m(\{x_1,x_2\})+\Delta = m(\{x_1,x_2\})+0.038461544 \cdot 4 \cdot 2$
\item $m'(\{x_1, \dots, x_5\})=m(\{x_1, \dots, x_5\}) = 0.07692307 \cdot 5$
\item $m'(\{x_1, \dots, x_8\})=0.07692307 \cdot 8 - \Delta = 0.07692307 \cdot 8  - 0.038461544 \cdot 8 = 0.03846153 \cdot 8$
\end{itemize}
The pignistic probability $P_{Bel'}$ corresponds to: 
\begin{itemize}
\item $P_{Bel'}(x_1)=P_{Bel'}(x_2)= 0.038461544 \cdot 4 = 0.15384617$
\item $P_{Bel'}(x_3)=P_{Bel'}(x_4)=P_{Bel'}(x_5)= 0.07692307 + 0.03846153 = 0.1153846$
\item $P_{Bel'}(x_6)=P_{Bel'}(x_7)=P_{Bel'}(x_8)= 0.03846153$
\end{itemize}

The entropy of $P_{Bel}$ is 2.0317593 and the entropy of $P_{Bel'}$ is 1.8300099. So, in this case transferring mass/evidence from a larger set to a smaller one reduces entropy. 
\end{example}

\begin{example}
\label{example.2.entropy}
Let $X=\{x_1, \dots, x_{10}\}$. Let $Bel$ be the belief function defined by: 
\begin{itemize}
\item $m(\{x_1, \dots, x_{10}\})= 0.08333332 \cdot 10$
\item $m(\{x_1, x_2\})=0.08333332 \cdot 2$
\end{itemize}
The pignistic probability $P_{Bel}$ corresponds to: 
\begin{itemize}
\item $P_{Bel}(x_1)=P_{Bel}(x_2)=0.16666664$
\item $P_{Bel}(x_3)=P_{Bel}(x_4)=P_{Bel}(x_5)\\
       =P_{Bel}(x_6)=P_{Bel}(x_7)=P_{Bel}(x_8)\\
       =P_{Bel}(x_9)=P_{Bel}(x_{10})=0.08333332$
\end{itemize}

Define $Bel'$ by {\em transferring} mass equivalent to $\Delta = 0.08333332 \cdot 10$ from $C_2=\{x_1, \dots, x_{10}\}$ to $C_1=\{x_3,\dots, x_{10}\}$. Then, the new belief function $Bel'$ is defined by: 
\begin{itemize}
\item $m'(\{x_1, \dots, x_{10}\})= m(\{x_1, \dots, x_{10}\}) - \Delta = 0$ %%  0.08333332 \cdot 10 - \Delta = 0$
\item $m'(\{x_3, \dots, x_{10}\})= m(\{x_3, \dots, x_{10}\}) + \Delta = 0 + 0.08333332 \cdot 10 \\
       = 0.10416665 \cdot 8$
\item $m'(\{x_1, x_2\})=m(\{x_1, x_2\})=0.08333332 \cdot 2$
\end{itemize}
Therefore, its pignistic probability $P_{Bel'}$ corresponds to: 
\begin{itemize}
\item $P_{Bel}(x_1)=P_{Bel}(x_2)= 0.08333332$
\item $P_{Bel}(x_3)=P_{Bel}(x_4)=P_{Bel}(x_5)\\
       =P_{Bel}(x_6)=P_{Bel}(x_7)=P_{Bel}(x_8)\\
       =P_{Bel}(x_9)=P_{Bel}(x_{10})=0.10416665$
\end{itemize}

Here, we have that the entropy of $P_{Bel}$ is 2.2538579 while the one of $P_{Bel'}$ is 2.2989538. So, we have that the entropy of the pignistic distribution of the belief function with more information $P_{Bel'}$ is larger than the entropy of the other distribution $P_{Bel}$. 
\end{example}

The behaviour of the entropy in these two examples can be explained from the fact that the entropy is a Schur-concave function (see e.g.~\cite{ref:Marshall.Olkin.Arnold.2011} for details). In the first example, $P_{Bel}$ majorizes $P_{Bel'}$, and therefore $entropy(P_{Bel}) \geq entropy(P_{Bel'})$. In the second example, is $P_{Bel'}$ who majorizes $P_{Bel}$ and thus $entropy(P_{Bel'}) \geq entropy(P_{Bel})$. 

We prove now that the measure of nonspecificity is monotonic with respect to a mass transfer. First we introduce this measure. 

\begin{definition} \cite{ref:Dubois.Prade.1985:nonspecificity}
The measure of non-specificity $N$ for a belief function $Bel$ is defined by
\begin{equation}
N(Bel)=\sum_{A \subseteq X} m(A) \log |A|
\end{equation}
\end{definition}

For this measure, the following holds. 

\begin{proposition}
\label{prop:Beliefs.nonspecificity}
Let $C_1$, $C_2$ be two subsets of $X$ such that $C_1 \subseteq C_2$, let $Bel$ be a belief function defined by $m$ and $Bel'$ a belief function defined by the following $m'$
\begin{itemize}
\item $m'(C_1)=m(C_1)+\Delta$
\item $m'(C_2)=m(C_2)-\Delta$
\item $m'(A)  =m(A)$  for all $A \neq C_1$ and $A \neq C_2$
\end{itemize}
where $\Delta$ is a value such that $m(A) \in [0,1]$ for all $A \subseteq X.$ 

Then, the nonspecificity of $Bel$ is larger than the nonspecificity of $Bel'$, that is
$$N(Bel) \geq N(Bel').$$
\end{proposition}

%%Let us first consider the nonspecificity of $Bel$:
%%\begin{array}{lll}
%%N(Bel) &= &\sum_{A \subseteq X} m(A) \log |A|\\
%%       &= &\sum_{A \subseteq X, A \neq C_1} m(A) \log |A|\\
%%       &  &+  m(C_1) \log |C_1|\\
%%       &  &+  m(C_2) \log |C_2|\\
%%\end{array}

\begin{proof}
To prove this proposition, let us consider the nonspecificity of $Bel$ and put it in terms of the nonspecificity of $Bel'$: 
\begin{equation}
\begin{array}{lll}
N(Bel')&= &\sum_{A \subseteq X} m'(A) \log |A|\\
       &= &\sum_{A \subseteq X, A \neq C_1} m'(A) \log |A|\\
       &   &+  m'(C_1) \log |C_1| \\
       &   &+  m'(C_2) \log |C_2|\\
       &= &\sum_{A \subseteq X, A \neq C_1} m(A) \log |A|\\
       &   &+  (m(C_1) + \Delta) \log |C_1| \\
       &   &+  (m(C_2) - \Delta) \log |C_2|\\
       &= &\sum_{A \subseteq X, A \neq C_1} m(A) \log |A|\\
       &   &+  m(C_1)  \log |C_1| \\
       &   &+  (m(C_2) \log |C_2|\\
       &   &+  \Delta \log |C_1| \\
	&  &-  \Delta \log |C_2|\\
       &= &N(Bel) +  \Delta (\log |C_1| - \log |C_2|)\\
\end{array}
\end{equation}
As $|C_1| < |C_2|$ we have that $\log |C_1| - \log |C_2| < 0$, so that $N(Bel) \geq N(Bel')$. 
\end{proof}

We have seen in this section that nonspecificity is useful to measure the information in a belief function, and that entropy is not. The {\em failure} of entropy in this context might seem to be in contradiction with the fact that entropy typically is understood as a measure of information. Nevertheless, the classification of uncertainty measures given in Klir and Wierman~\cite{ref:Klir.Wierman.1999} sheds some light over this issue. Entropy is classified as a measure of conflict. In this sense, the belief function $m'$ in Example~\ref{example.2.entropy} presents a larger conflict for a decision than when mass is transferred from $C_2$ to $C_1$. This is so because the probabilities are much more similar in $m'$ that in $m$ although we have more information in $m'$ than in $m$. In contrast, Klir and Wierman classify non-specificity as a measure of imprecision, which is said to be connected with sizes (cardinalities) of relevant sets of alternatives (see p. 43 in~\cite{ref:Klir.Wierman.1999}) which is precisely the case here. Proposition~\ref{prop:Beliefs.nonspecificity} clearly shows that any transference of mass from a set to a more concrete one will always increase the measure. Indeed, we have that $N(Bel)=0$ when $Bel$ is a probability distribution. 

\subsection{Conditioning}
When additional information is given to the re-identification algorithm, the belief function is expected to change accordingly. In general, we can pressume that this information can also be represented in terms of a belief function. Note that conditioning in probability theory can be expressed by a belief function. For example, conditioning by the presence of an element in a set $B$ which in probability theory results into the probability $P(A|B)$, will be expressed by the conditioning with respect to the belief function generated from $m_B(B)=1$ and $m_B(A)=0$ for all $A \neq B$. Naturally, the belief function used in the conditioning should also be compatible with the true probability. 

Therefore, given two belief functions compatible with the true probability, the conditioning should lead to another belief function, also compatible with the true probability. 

\begin{definition}
Given two belief functions $Bel_1$ and $Bel_2$ compatible with a true probabilty $P$, an {\em acceptable combination function} ${\mathbb{C}}$ is a combination function that returns a new belief function $Bel$ that is compatible with $P$ and such that $N(Bel) \leq \min (N(Bel_1),N(Bel_2))$. 
\end{definition}

Any combination function that satisfies this property will be suitable for conditioning. See e.g.~\cite{{ref:Walley.1991},{ref:Torra.1995}} for functions satisfying this property. 

\begin{definition}
\label{def:combinacio.r.beliefs}
Given $r$ belief functions $Bel_1, \dots, Bel_r$ compatible with a true probability $P$, we define their combination ${\mathbb{C}'}$ as the extension of the acceptable combination function ${\mathbb{C}}$ as follows: 
\begin{equation}
{\mathbb{C}'}(Bel_1, \dots, Bel_r)=\left\{
	\begin{array}{ll}
	  {\mathbb{C}}({\mathbb{C}'}(Bel_1, \dots, Bel_{r-1}), Bel_i) & \textrm{if } i>2 \\
	  {\mathbb{C}}(Bel_1,Bel_2)                                   & \textrm{if } i=2 \\
	\end{array}
   \right.
\end{equation}
\end{definition}

%%\begin{example}
%%{\bf COMBINACIO DE DUES MESURES.}
%%\end{example}

Then, when different items $it_1, \dots, it_k$ of additional information are considered, all of them expressed by means of belief functions $Bel_{it_1}, \dots, Bel_{it_k}$, which are compatible with the true probability, and we combine them, the result will converge to be the true probability, and the nonspecificity is reduced. When the true probability is achieved, we have that the nonspecificity is zero. 

\begin{proposition}
Let $it_1, \dots, it_k$ be a set of items of additional information expressed by means of belief functions $Bel_{it_1}, \dots, Bel_{it_k}$ compatible with a true probability $P$. Let ${\mathbb{C}}$ be an acceptable combination function. 

Then, the combination of belief functions $Bel_{it_1} \dots Bel_{it_r}$ for $r < k$ using ${\mathbb{C}'}$ as in Definition~\ref{def:combinacio.r.beliefs} is a belief function $Bel_r$ compatible with the true probability $P$, and such that $N(Bel_r) \geq N(Bel_{r+1})$. 

In addition, if, for a given $r_0$ the belief function $Bel_{r_0}$ is a probability, then $Bel_{r_0}=P$ and for all $r \geq r_0$ we have $Bel_{r_0} = Bel_{r}$. 
\end{proposition}

The proof of this proposition is trivial taking into account that ${\mathbb{C}}$ is an acceptable combination function and that $N(Bel)=0$ when $Bel$ is a probability distribution. 

\section{Conclusions}
In this paper we have formalized re-identification algorithms in terms of belief functions and probabilities. We have shown that belief functions and their pignistic transformation permits us to express the uncertainty in re-identification algorithms in a natural way. 

\section*{Acknowledgements}
Partial support by the Spanish MEC projects ARES (CONSOLIDER INGENIO 2010 CSD2007-00004), eAEGIS (TSI2007-65406-C03-02), COPRIVACY (TIN2011-27076-C03-03), and RIPUP (TIN2009-11689) is acknowledged. Partial support of the European Project DwB (Grant Agreement Number 262608) is also acknowledged. One author is partially supported by the FPU grant (BOEs 17/11/2009 and 11/10/2010) and by the Government of Catalonia under grant 2009 SGR 1135. The authors are with the UNESCO Chair in Data Privacy, but their views do not necessarily reflect those of UNESCO nor commit that organization.


\begin{thebibliography}{99}
\bibitem{ref:Abril.Navarro-Arribas.Torra.2012:IF}
Abril, D., Navarro-Arribas, G., Torra, V. (2012) Improving record linkage with supervised learning for disclosure risk assessment, Information Fusion, in press.

\bibitem{ref:Chateauneuf.1994}
Chateauneuf, A. (1994) Combination of compatible belief functions and relation of specificity, in Advances in Dempster-Shafer Theory of evidence, Wiley, 97-114. 

\bibitem{ref:Copas.Hilton.1990}
Copas, J. B., Hilton, F. J. (1990) Record linkage: statistical models for matching computer records, J. R. Statist. Soc. A 153  287-320. 

\bibitem{ref:Dubois.Prade.1985:nonspecificity}
Dubois, D., Prade, H. (1985) A note on measures of specificity for fuzzy sets, Int. J. of General Systems 10:4 279-283. 
%% nonspecificity

\bibitem{ref:Fellegi.Sunter.1969}
Fellegi, I. P., Sunter, A. B. (1969) A theory of record linkage, Journal of the American Statistical Association 64 1183-1210. 

\bibitem{ref:Higashi.Klir.1983}
Higashi, M., Klir, G. J. (1983) Measures of uncertainty and information based on possibility distributions, Int. J. of General Systems 9:1 43-58. 

\bibitem{ref:Klir.Wierman.1999}
Klir, G. J., Wierman, M. J. (1999) Uncertainty-Based Information: Elements of Generalized Information Theory, Physica-Verlag. 
%% mesures de nonspecificity i altres mesures que generalitzen les entropies. 

\bibitem{ref:Marshall.Olkin.Arnold.2011}
Marshall, A. W., Olkin, I., Arnold, B. C. (2011) Inequalities: Theory of Majorization And Its Applications, Springer (2nd edition). 
%% http://www.springer.com/statistics/statistical+theory+and+methods/book/978-0-387-40087-7

\bibitem{ref:Nin.Herranz.Torra.2008:DKE.rankSwapping}
Nin, J., Herranz, J., Torra, V. (2008) Rethinking Rank Swapping to Decrease Disclosure Risk, Data and Knowledge Engineering, 64:1 346-364.

\bibitem{ref:Ramer.1987}
Ramer, A. (1987) Uniqueness of information measure in the theory of evidence, Fuzzy Sets and Systems 24:2 183-196. 

\bibitem{ref:Shafer.1976}
Shafer, G. (1976) A Mathematical Theory of Evidence, Princeton University Press, Princeton, New Jersey.

\bibitem{ref:Smets.Kennes.1994}
Smets, P., Kennes, R. (1994) The transferable belief model, Artificial Intelligence 66 191-234. 

\bibitem{ref:Stokes.Torra.2012:PAIS}
Stokes, K., Torra, V. (2012) n-Confusion: a generalization of k-anonymity, Proc. of the PAIS workshop, Berlin, Germany, 30 March, 2012, in press. 

\bibitem{ref:Stokes.Torra.2012:SOCO}
Stokes, K., Torra, V. (2012) Reidentification and k-anonymity: a model for disclosure risk in graphs, Soft Computing, in press. 

\bibitem{ref:Torra.1995}
Torra, V. (1995) A New combination function in evidence theory, Intl. J. of Intel. Syst., 10:12 1021-1033. 

\bibitem{ref:Torra.Abowd.Domingo.2006:PSD}
Torra, V., Abowd, J. M., Domingo-Ferrer, J. (2006) Using Mahalanobis Distance-Based Record Linkage for Disclosure Risk Assessment, Lecture Notes in Computer Science 4302 233-242.
%% Mahalanobis distance

\bibitem{ref:Walley.1991}
Walley, P. (1991) Statistical reasoning with imprecise probabilities, Chapman and Hall. 

\bibitem{ref:Winkler.2004:PSD}
Winkler, W. E. (2004) Re-identification methods for masked microdata, PSD 2004, Lecture Notes in Computer Science 3050 216-230.

\bibitem{ref:Yancey.Winkler.Creecy.2002}
Yancey, W. E., Winkler, W. E., Creecy, R. H. (2002) Disclosure risk assessment in perturbative microdata protection, in J. Domingo-Ferrer (ed.) Inference Control in Statistical Databases, Lecture Notes in Computer Science 2316 135-152. 

\end{thebibliography}
\end{document}